\documentclass[sigconf,nonacm,natbib=false,balance=false]{acmart}
\setcopyright{none}
\usepackage[T1]{fontenc}
\usepackage{mathtools}
\usepackage[linesnumbered,lined,boxed,commentsnumbered]{algorithm2e}
\usepackage{hyperref}
\usepackage{enumitem}
\usepackage[backend=biber,style=acmnumeric,sorting=nyt]{biblatex}
\theoremstyle{definition}
\newtheorem{definition}{Definition}

\theoremstyle{plain}
\newtheorem{theorem}{Theorem}

\newtheorem{lemma}{Lemma}

\theoremstyle{remark}

\newcommand{\proba}[1]{\mathbb{P}\left[#1\right]}
\newcommand{\E}[2][]{\mathbb{E}_{#1} \left[ #2 \right]}

\newcommand{\bigo}[1]{O\left( #1 \right)}

\SetKwProg{Fn}{Function}{}{}
\SetKwFunction{Hypergeometric}{Hypergeometric}
\SetKwFunction{RR}{RR}
\SetKwFunction{Lap}{Lap}
\SetKwFunction{PPR}{PPR}
\SetKwFunction{Gen}{Gen}
\SetKwFunction{ProbaRatio}{ProbabilityRatio}
\SetKwFunction{EncodeRR}{EncodeRR}
\SetKwFunction{DecodeRR}{DecodeRR}

\begin{document}

%%
%% The "title" command has an optional parameter,
%% allowing the author to define a "short title" to be used in page headers.
\title{Communication-Efficient Publication of Sparse Vectors under Differential Privacy}

%%
%% The "author" command and its associated commands are used to define
%% the authors and their affiliations.
%% Of note is the shared affiliation of the first two authors, and the
%% "authornote" and "authornotemark" commands
%% used to denote shared contribution to the research.
\author{Quentin Hillebrand}
\orcid{0000-0002-7747-4998}
\affiliation{%
    \institution{The University of Tokyo}
    \city{Tokyo}
    \country{Japan}}
\email{quentin-hillebrand@g.ecc.u-tokyo.ac.jp}

\author{Vorapong Suppakitpaisarn}
\orcid{0000-0002-7020-395X}
\affiliation{%
    \institution{The University of Tokyo}
    \city{Tokyo}
    \country{Japan}}
\email{vorapong@is.s.u-tokyo.ac.jp}

\author{Tetsuo Shibuya}
\orcid{0000-0003-1514-5766}
\affiliation{%
    \institution{The University of Tokyo}
    \city{Tokyo}
    \country{Japan}}
\email{tshibuya@hgc.jp}

%%
%% By default, the full list of authors will be used in the page
%% headers. Often, this list is too long, and will overlap
%% other information printed in the page headers. This command allows
%% the author to define a more concise list
%% of authors' names for this purpose.
% \renewcommand{\shortauthors}{Hillebrand, Suppakitpaisarn, Shibuya}

%%
%% The abstract is a short summary of the work to be presented in the
%% article.
\begin{abstract}
In this work, we propose a differentially private algorithm for publishing matrices aggregated from sparse vectors. These matrices include social network adjacency matrices, user-item interaction matrices in recommendation systems, and single nucleotide polymorphisms (SNPs) in DNA data. Traditionally, differential privacy in vector collection relies on randomized response, but this approach incurs high communication costs. Specifically, for a matrix with $N$ users, $n$ columns, and $m$ nonzero elements, conventional methods require $\Omega(n \times N)$ communication, making them impractical for large-scale data. Our algorithm significantly reduces this cost to $O(\varepsilon m)$, where $\varepsilon$ is the privacy budget. Notably, this is even lower than the non-private case, which requires $\Omega(m \log n)$ communication. Moreover, as the privacy budget decreases, communication cost further reduces, enabling better privacy with improved efficiency. We theoretically prove that our method yields results identical to those of randomized response, and experimental evaluations confirm its effectiveness in terms of accuracy, communication efficiency, and computational complexity.
\end{abstract}

%%
%% The code below is generated by the tool at http://dl.acm.org/ccs.cfm.
%% Please copy and paste the code instead of the example below.
%%


%%
%% Keywords. The author(s) should pick words that accurately describe
%% the work being presented. Separate the keywords with commas.

\received{20 February 2007}
\received[revised]{12 March 2009}
\received[accepted]{5 June 2009}

\maketitle

\section{Introduction}

\textbf{Differential privacy} \cite{dwork2006calibrating} has emerged as a widely accepted standard for protecting sensitive information while enabling data analysis and sharing. By adding carefully calibrated noise to query results, it ensures that the release of analysis outcomes does not significantly alter an observer's knowledge of users' sensitive information.

To handle cases where data is dispersed across multiple users, various techniques have been developed to generate an obfuscated representation of the raw data \cite{erlingsson_rappor_2014}. 
One of the most well-known techniques is \textbf{randomized response}.
In this technique, each data point is probabilistically altered before being transmitted to the central server. The server then processes the received obfuscated data rather than the original values. This technique is widely recognized for its ability to satisfy local differential privacy \cite{kasiviswanathan_what_2011}, a variant of differential privacy that not only safeguards users' information when publishing analytical results but also ensures privacy during data transmission and storage on the central server.

When a user possesses multiple sensitive data points, a variant of differential privacy known as \textbf{metric differential privacy} \cite{andres2013geo} is commonly employed to enhance privacy protection. This privacy framework enables more accurate data analysis while maintaining user confidentiality. The randomized response technique, which perturbs each data point with a certain probability, also adheres to the requirements of this privacy notion.

We consider the scenario where each user's data is represented as a high-dimensional yet sparse vector—that is, while a user may have a large amount of information, most of the values are identical. For instance, in the context of \textbf{graph and social network data}, each user maintains a list of friends, which can be represented as an adjacency vector of size \( n \), where \( n \) is the total number of users in the system. If a user has only \( d \ll n \) friends, the adjacency vector contains just \( d \) ones, while the remaining entries, which are zeros (the trivial values in this case), make up nearly the entire vector. A similar situation arises in \textbf{recommendation systems}, where each user provides ratings for movies. Since users typically rate only a small subset of all available movies, the majority of entries in their rating vector remain as "N/A," indicating missing or unrated values. In genomic data, it is well established that the vast majority of \textbf{genetic information} is identical across all humans, with differences between individuals accounting for only about 1\% of the total genetic sequence.

In a non-private setting, collecting and storing sparse vectors is efficient because we only need to track the positions and values of non-trivial elements. Given a vector of size \( n \) with at most \( d \) non-trivial entries, the communication and storage cost is \( O(d \log n) \). However, when applying the randomized response mechanism, each entry is obfuscated with a certain probability, increasing the number of non-trivial elements to \( \Omega(n) \) after obfuscation. As a result, the communication and storage cost per user grows to \( \Omega(n) \) bits. 
For a system with \( N \) users, the overall communication and storage complexity becomes \( \Omega(n \times N) \) bits, which is impractical for many applications. 

For example, in a social network setting where $n = N$, this results in a communication and storage cost of \( \Omega(N^2) \) bits, making it impractical to process networks with more than approximately \( N = 50,000 \) nodes. Moreover, some algorithms \cite{imola2021locally, hillebrand2024cycle} require the central server to distribute the collected results back to all users, further increasing the overall communication cost to \( \Omega(N^3) \), which is prohibitively expensive.

Reducing communication costs has been a focus of many recent studies, including approaches such as rejection sampling \cite{feldman2021lossless} and importance sampling \cite{shah2022optimal}. 
However, the outputs of these techniques do not precisely match the distribution of the original mechanism.
This discrepancy can be problematic, as it may cause the mechanism to lose essential properties such as unbiasedness, making it more challenging to analyze and post-process. 

Several mechanisms provide unbiased results \cite{imola2022communication, hillebrand2023communication, liu2024universal}. However, the methods proposed in \cite{imola2022communication, hillebrand2023communication} are specifically designed for social network data and still require a communication cost of \( \Omega(N^2) \). While the mechanism in \cite{liu2024universal} is universally applicable, its use for high-dimensional vectors incurs a communication cost of \( \Omega(n \times N) \) and suffers from significant computational overhead. In particular, the computational cost grows exponentially with \( d \), making it impractical even when the number of non-trivial entries is below 10,000, as the computation time becomes prohibitively high.

\subsection{Our Contributions}

In this paper, we propose a communication- and computation-efficient mechanism for publishing the randomized response of high-dimensional sparse vectors. 

Our mechanism incorporates some ideas from the approach called Poisson Private Representation (PPR) in~\cite{liu2024universal}. However, to mitigate the prohibitive computational cost that arises when the number of non-trivial elements \( d \) is large, we introduce a random partitioning strategy. Specifically, we divide the high-dimensional vector into \( \Theta(d) \) chunks, ensuring that the expected number of non-trivial elements per chunk is $\Theta(1)$. This approach minimizes the likelihood of any chunk containing a disproportionately large number of non-trivial elements, which causes high computation cost in \cite{liu2024universal}. We then propose an algorithm to efficiently compress the randomized response results for vectors with a small number of non-trivial elements. The complete mechanism is presented in Section \ref{sec:crr}.

In \cite{liu2024universal}, the algorithm known as chunk PPR partitions the input vector into \( d \) chunks. However, this algorithm is designed for general inputs rather than specifically for sparse vectors. While the technique effectively reduces the computation time of PPR, our division method leverages the sparsity of the vector more efficiently.

In Section~\ref{sec:theory}, we also provide a theoretical analysis of our algorithm, demonstrating that both the communication cost and execution time are independent of the vector's length. Specifically, given a privacy budget \( \varepsilon \) and \( d \) non-zero elements in the vector, we show that the communication cost is bounded by \( \bigo{\varepsilon d} \), while the computational complexity remains \( \bigo{d} \). 

We emphasize that our communication cost depends only on the number of non-trivial elements, \( d \), rather than the size of the vector, \( n \). As a result, our cost is significantly lower than any previous algorithm designed for high-dimensional sparse vectors. Notably, our approach even outperforms the non-private case, which incurs a cost of \(\Omega(d \log n)\). 

Furthermore, since a smaller privacy budget corresponds to stronger privacy guarantees, our method achieves lower communication costs while ensuring better privacy protection. This contrasts with prior works on social network differential privacy \cite{imola2022communication,hillebrand2023communication}, where the communication cost increases as privacy improves. In those approaches, a smaller \(\varepsilon\) leads to a larger number of non-trivial elements in the obfuscated vector, thereby increasing the communication overhead.

We also emphasize that the expected computational cost of our algorithm scales linearly with the number of non-trivial elements, \( d \). This represents a significant improvement over the mechanism proposed in \cite{liu2024universal}, where the computation time grows exponentially with \( d \).

From the compressed data, each element of the randomized response result can be retrieved in constant time.

Some applications of our algorithm are listed below. We also give the details of these application in Section~\ref{sec:applications}.

\paragraph{Graph/Social Network Data} The private publication of social network adjacency lists is a fundamental tool in various graph privacy tasks. For instance, it serves as the first step in many synthetic graph generation frameworks; see \cite{liu2024pgb} for a benchmark.   Additionally, it plays a key role in subgraph counting, forming the initial phase of the general two-step mechanism \cite{imola2021locally}, which was originally designed for triangle counting but has since been extended to cycle counting \cite{hillebrand2024cycle} and common neighbor estimation \cite{he2024common}.  

Given that this framework requires transmitting complete adjacency data to each user, reducing communication overhead is a crucial consideration. Our method allows users to publish their adjacency lists under edge-local differential privacy \cite{qin2017generating}, with a communication cost proportional to their degree.  
\paragraph{Recommendation Systems} 
A comprehensive survey on private recommendation systems can be found in \cite{himeur2022latest}, highlighting the growing interest in this field.  
In \cite{berlioz2015applying}, the authors propose perturbing local user interaction ratings before submitting them to a central server that generates recommendations. However, their privacy model differs from ours, as it only protects the rating values while leaving their existence or absence unprotected.  
In \cite{gao2020dplcf}, randomized response is employed for local obfuscation before transmitting user-item interactions to the central server.

We enable the publication of user-item interaction matrices, where the communication cost scales with the number of ratings. 
\paragraph{Genomic Data} A common task in private genomics is the publication of single nucleotide polymorphisms (SNPs), which provide valuable insights into DNA variations and help assess genomic risk factors. Several methods have been developed for this purpose \cite{humbert2014reconciling, yilmaz2020preserving, yilmaz2022genomic}, though they rely on different privacy definitions than ours. In \cite{yamamoto2024privacy}, randomized response is employed to publish SNP data under the \(\varepsilon\)-DP privacy model.

Our approach supports the release of single nucleotide polymorphisms (SNPs), with a communication cost proportional to the number of locations where the least common variation is present.  

In Section~\ref{sec:experiments}, we present experimental results across all three application scenarios, highlighting the effectiveness and practicality of our mechanism. We verify that our communication and computation costs remain minimal even for vectors as large as 890,060 in size. Additionally, we confirm that a graph algorithm for triangle counting based on our framework achieves 100 times lower communication cost than any previous method, even in a small social network. With the same communication cost, our precision is $10^4$ better than any previous works.

\section{Preliminaries}
\label{sec:preliminaries}

\subsection{Sampling with and without Replacement}

%Sampling is a fundamental concept in probabilistic analysis, typically categorized into two main methods: sampling with replacement and sampling without replacement.  

In \textbf{sampling with replacement}, each draw is made independently from the entire population. If \( n \) draws are performed, each with a probability of success \( p \), the number of successes, denoted by \( X \), follows a binomial distribution \( \mathcal{B}(n, p) \). The probability of obtaining exactly \( k \) successes is given by 
$\mathbb{P}(X = k) = \binom{n}{k} p^k (1 - p)^{n - k}$.  

In contrast, \textbf{sampling without replacement} means that each draw is taken from the remaining unselected portion of the population. For a population of size \( N \) containing \( K \) successful elements, when drawing \( n \) times, the number of successes, denoted by \( Y \), follows a hypergeometric distribution \( \Hypergeometric(N, K, n) \). The probability of observing exactly \( k \) successes is given by   
$\mathbb{P}(Y = k) = \frac{\binom{K}{k} \binom{N - K}{n - k}}{\binom{N}{n}}$.

%Although analyzing sampling without replacement is generally more complex, under certain conditions, it can be effectively approximated by sampling with replacement.
We will use the following theorem in our analysis.
\begin{theorem}[Theorem 4 of \cite{hoeffding1963probability}]
    \label{thm:hypergeometric-inequality}
Let \( f \) be a continuous and convex function. If \( Y \sim \Hypergeometric(N, K, n) \) and \( X \sim \mathcal{B}(n, p) \) with \( p = K / N \), then  
$\mathbb{E}[f(X)] \leq \mathbb{E}[f(Y)]$.
\end{theorem}
Additionally, the standard formula for the moment generating function of \( X \sim \mathcal{B}(n, p) \) is given as $\mathbb{E}\left[e^{tX}\right] = \left(1 - p + p e^t\right)^n$ (see Chapter 7.7 of \cite{ross1976first}).

\subsection{Counter-based Generators}

%Unlike sequential pseudo-random number generators (PRNGs), 
Counter-based pseudo-random number generators (PRNGs) allow for the parallel generation of pseudo-random number sequences.  
Rather than generating all numbers from index 0 to \( i-1 \) before obtaining the number at index \( i \), each index can be generated independently without extra computational cost.  
This property is formally stated in Theorem \ref{thm:counter-based}.

\begin{theorem}[Counter-based PRNGs \cite{salmon2011parallel}] 
    \label{thm:counter-based}  
    There exist PRNGs \( f \) such that, given a key \( \mathcal{K} \) and an index \( i \), the \(i\)-th element of the random sequence \( f(\mathcal{K}, i) \) can be generated in constant time.  
\end{theorem}

There exist several efficient implementations of counter-based PRNGs (e.g., \cite{salmon2011parallel}), but in this article, we treat them as black-box mechanisms to simplify the presentation of our methods.  

To formalize our notation, let \(\mathcal{G}\) be an initialized counter-based PRNG. We denote by \(\mathcal{G}^{(k)}\) the state of the generator after \( k \) draws. Furthermore, for a given distribution \( Q \), we define \(\Gen(Q, \mathcal{G}^{(k)})\) as a sample drawn from \( Q \) using the randomness generated by \(\mathcal{G}\) at state \( \mathcal{G}^{(k)} \).

\subsection{Differential Privacy}

%We formally define the concept of metric differential privacy in the following definition.
\begin{definition}[Metric differential privacy \cite{andres2013geo}]
    Let $\mathsf{d}$ be a distance between datasets. A mechanism $\mathcal{M}$ satisfies $\varepsilon$-metric differential privacy if for any pair of datasets $D$ and $D'$, and for any possible outcome $S$ of $\mathcal{M}$, we have
        $\proba{\mathcal{M}(D) \in S} \leq e^{\varepsilon d(D, D')} \proba{\mathcal{M}(D') \in S}$.
\end{definition}
The parameter $\varepsilon$ is called the privacy budget of $\mathcal{M}$.

Metric differential privacy was initially introduced as a convenient way to define privacy in metric spaces. However, it has also proven to be a highly general framework, as both the classic notions of local and central differential privacy can be seen as special cases of metric differential privacy under the appropriate choice of distance.

In this article, we consider an input space of \( n \)-dimensional vectors and define privacy using the Hamming distance \( \mathcal{H} \). Specifically, for two vectors \( v \) and \( v' \), the distance \( \mathcal{H}(v, v') \) is equal to the number of coordinates on which they differ.

The algorithm we will analyze in this article is randomized response, a method used to privately publish an entire vector of categorical data.
\begin{definition}[Randomized response \cite{wang2016using}]
For \( \varepsilon > 0 \) and a vector \( v \) where each element of \( v \), denoted by $v_1, \dots, v_n$ belongs to \( \{0, \ldots, k-1\} \), the randomized response mechanism \( \mathcal{R} \) with privacy budget \( \varepsilon \) is defined as follows. For any \( j \in \{0, \ldots, k-1\} \), the probability of reporting \( j \) instead of the true value \( v_i \) is given by:
\[
\mathbb{P}(\mathcal{R}(v_i) = j) =
\begin{cases}
\frac{e^{\varepsilon}}{e^{\varepsilon} + k - 1} & \text{if } v_i = j, \\
\frac{1}{e^{\varepsilon} + k - 1} & \text{otherwise}.
\end{cases}
\]
\end{definition}

With the Hamming distance, randomized response satisfies \( \varepsilon \)-metric differential privacy.

\subsection{Poisson Private Representation}

In \cite{liu2024universal}, the authors present a method called Poisson Private Representation (PPR) for converting any differentially private mechanism into a compressed version while preserving the original output distribution \( P \). The transformed algorithm ensures an identical distribution to the original mechanism while having a communication cost of \( \bigo{\varepsilon} \).

This method leverages the shared random number generator results between the user and the central server, represented by draws \((Z_i)_{i \in \mathbb{N}}\) from a candidate distribution \( Q \). The server should select \( Q \) to closely approximate the true output distribution \( P \) of the differentially private (DP) mechanism applied to the private data. While the server does not have access to the private data—and therefore cannot directly determine \( P \)—a practical approach is to use the output distribution of the same DP mechanism on arbitrary input data as an approximation.

With these shared draws, the user can transmit the index \( K \) corresponding to the selected draw, which serves as the mechanism's output. The central server can then retrieve the output by computing \( Z_K \). Notably, if counter-based PRNGs are used, this computation remains constant-time regardless of the value of \( K \). 

Theorem \ref{thm:ppr-privacy} establishes that this mechanism preserves both privacy and the original output distribution.
\begin{theorem}[Proposition 4.2 and Theorem 4.7 of \cite{liu2024universal}]
    \label{thm:ppr-privacy}
For an \(\varepsilon\)-metric differentially private mechanism \(\mathcal{M}\), the PPR simulation of \(\mathcal{M}\) with parameter \(\alpha > 1\) satisfies \(2\alpha\varepsilon\)-metric differential privacy while ensuring that its output follows the same distribution as that of \(\mathcal{M}\).
\end{theorem}

Two types of communication occur during the protocol. The first is dedicated to establishing shared randomness between the user and the central server. This can be achieved, for instance, by transmitting a public instantiation key or using a predefined mechanism. Since this step can be initiated by either party and incurs only a small, constant communication cost, we exclude it from our communication analysis.

The second type of communication, which is our primary focus, involves the local user transmitting the value of \( K \) to the central server. This integer can be efficiently encoded using Huffman coding \cite{huffman1952method,li2018strong}, resulting in an expected communication cost of at most \( \mathbb{E}[\log_2 K] + \log_2 (\mathbb{E}[\log_2 K] + 1) + 2 \) bits. Consequently, bounding the expected communication cost reduces to the problem of bounding \( \mathbb{E}[\log_2 K] \).

\begin{theorem}[Theorem 4.3 of \cite{liu2024universal}]
    \label{thm:ppr-communication}
For a mechanism with output distribution \( P \), a candidate distribution \( Q \), and a parameter \( \alpha > 1 \), the message \( K \) produced by PPR applied to \( P \) satisfies  
$\E{\log_2 K} \leq \mathsf{D}(P \| Q) + \frac{\log_2 3.56}{\min\{(\alpha -1)/2, 1\}}$,
where \( \mathsf{D}(P \| Q) \) denotes the KL-divergence between \( P \) and \( Q \).
\end{theorem}

Another important consideration is the computational cost of the mechanism. On the server side, assuming the use of counter-based PRNGs, this cost remains constant and minimal.  
On the user side, however, the mechanism requires evaluating \(\frac{\mathrm{d}P}{\mathrm{d}Q}(Z_i)\), which represents the ratio of the probability of obtaining \( Z_i \) from \( P \) to that from \( Q \), for a large number of points \( Z_i \) sampled from the probability distribution \( Q \). To optimize this process, \cite{liu2024universal} introduces a reparameterization trick that employs a heap-based algorithm, ensuring the number of draws is bounded by \( \bigo{\sup_z \frac{\mathrm{d}P}{\mathrm{d}Q}(z)} \). Given that computing the probability ratio \(\frac{\mathrm{d}P}{\mathrm{d}Q}(z)\) incurs a cost of \( c \), the overall runtime of the PPR algorithm is  
$\bigo{\sup_z \frac{\mathrm{d}P}{\mathrm{d}Q}(z) \cdot \left( c + \log \sup_z \frac{\mathrm{d}P}{\mathrm{d}Q}(z) \right)}$.

\section{Our Algorithm: Compression of Randomized Response}
\label{sec:crr}

%In this section, we explore how the PPR mechanism can be used to publish vectors via randomized response while preserving metric differential privacy. %We begin by demonstrating that a naive application of PPR yields an inefficient algorithm. Then, we follow by an introduction to more refined techniques that improve performance.

The PPR mechanism can simulate any differentially private mechanism while reducing communication costs. Since randomized response is a differentially private mechanism, PPR can be directly applied to it. However, directly using PPR on the output of the randomized response method presents several challenges.

\paragraph{Issue of PPR: Large Hamming Distance}
Since two distinct vectors of length \( n \) can have a Hamming distance of up to \( n \), the KL divergence between the candidate distribution and the actual distribution can reach \( n\varepsilon \). Given this and considering Theorem~\ref{thm:ppr-communication}, we can infer that the resulting communication cost scales as \( \bigo{n\varepsilon} \). Consequently, directly applying PPR to randomized response within this privacy model does not yield any improvement over naive randomized response.

The direct usage of PPR is challenging even when we focus on the common scenario where vectors are sparse or close to a reference. In this setting, the general structure of the vector is largely known before the user’s data is published, with only a small number of coordinates differing from the reference. However, since the indices of these differing coordinates are unknown, the entire vector must still be published.  

\paragraph{Our Idea 1: Selection of the Candidate Distribution} Let \( d \) denote the number of differing coordinates from the reference. We select candidate distribution $Q$ as the randomized response applied to the reference vector. For instance, when the vector \( v \in \{0,1\}^n \) to be published is sparse, with only \( d \) nonzero elements while the rest are zeros, the candidate distribution \( Q \) is chosen as the distribution of the randomized response applied to the all-zero vector. 

Under this choice, the KL divergence between the candidate distribution and the randomized response applied to the actual vector is reduced to \( d\varepsilon \). By carefully selecting the candidate distribution, we achieve a communication cost proportional to the vector’s sparsity, similar to the non-private case.

\paragraph{Issue of PPR: Computation Cost} Since the number of draws (denoted as \( \sup_z \frac{\mathrm{d}P}{\mathrm{d}Q}(z))  \) in Section \ref{sec:preliminaries}) increases exponentially with \( d\varepsilon \), the computational cost of PPR also scales exponentially when applied to the randomized response result (see Section 8 of \cite{liu2024universal}). This makes it impractical for reasonable values of \( d \). To reduce the number of required draws, we partition the adjacency vector into smaller groups. While the original PPR paper also employs chunking, their method cannot be directly applied here to ensure low computational cost. Their approach relies on contiguous chunks, but depending on the data structure, values may be concentrated in specific regions of the vector. 

\paragraph{Our Idea 2: Random Partitioning Strategy} To reduce the number of draws, we divides the vector into smaller chunks using random partitioning strategy, ensuring that each chunk has a small degree on average. PPR is then applied independently to each chunk. By carefully selecting the chunk size, we achieve a communication cost proportional to the degree while maintaining a computational cost that also scales with the degree.

\begin{algorithm}
    \caption{Encodes the compressed randomized response of a vector}
    \label{algo:encode-rr}
    
    \Fn{\EncodeRR}{
        \KwIn{1) A list of indices corresponding to the non-trivial entries of the input vector, denoted as $(x_1, \ldots, x_d)$; 2) the values of the input vector at these indexes $(v_1, \ldots, v_d)$; 3) The reference vector values at these indexes $(c_1, \ldots, c_d)$; 4) The number of chunks $m$; 5) A privacy budget $\varepsilon$; 6) A public random permutation function $\varphi$}
        \KwOut{The encoded result $(K_1, \dots, K_m)$}
$y_1, \ldots, y_d \gets \varphi(x_1), \ldots, \varphi(x_d)$\;
$s \gets \lceil n / m \rceil$\;
For each $i \in [1, d]$, compute the Euclidean division of $y_i$ by $s$, yielding $(q_i, r_i)$\;
\For{$j \gets 1$ \KwTo $m$}{
    $S_j \gets \{(r_i, u_i, c_i) \mid q_i = j\}$\;
    $K_j \gets \PPR(S_j, \varepsilon, \alpha)$\;
}
\Return $(K_1, \ldots, K_m)$
    }
\end{algorithm}

Our encoding algorithm is given in Algorithm \ref{algo:encode-rr}. It begins by randomly permuting the vector using the function $\varphi$, and divide it into \( m \) chunks. The non-trivial elements of those $m$ chunks are denoted by the set \( S_1, \dots, S_m \). Next, we apply the PPR method \cite{liu2024universal} to each chunk and return the resulting list as the output.

Our decoding algorithm is presented in Algorithm \ref{algo:decode-rr}. To determine the value of the input vector at index \( i \), we first compute the chunk number \( q \) and the position of \( i \) within the chunk, denoted by \( r \), using the public function \( \varphi \) and \( s \). We then perform the decoding using the same approach as PPR.

\begin{algorithm}
    \caption{Decodes the compressed randomized response of a vector}
    \label{algo:decode-rr}
    
    \Fn{\DecodeRR}{
        \KwIn{A list of compressed indexes $(K_1, \ldots, K_m)$, a list of distributions $(Q_1, \ldots, Q_n)$, the index $i$ that one wants to access}
        \KwOut{The value of the vector at index $i$}

        $j \gets \varphi(i)$\;
        Let $(q, r)$ be the result of the Euclidean division of $j$ by $s$\;
        \Return $\Gen(Q_i, \mathcal{G}^{(s K_q +r)})$
    }
\end{algorithm}

\paragraph{Our Idea 3: Efficient Calculation of $\frac{\mathrm{d}P}{\mathrm{d}Q}(Z_i)$} As noted in Section \ref{sec:preliminaries}, PPR requires multiple evaluations of \(\frac{\mathrm{d}P}{\mathrm{d}Q}(Z_i)\). A naive approach to computing \(\frac{\mathrm{d}P}{\mathrm{d}Q}(Z_i)\) involves generating the entire randomized vector and comparing its probability under both distributions. Even when Algorithm \ref{algo:encode-rr} reduces the vector size from \( n \) to \( s \), \( s \) typically remains of the same order as \( n \), making this method computationally inefficient and leading to significant computation time.

In Algorithm~\ref{algo:ratio}, we propose an efficient approach leverages the fact that this probability ratio depends only on the values of the draw at the indices where the private vector differs from the reference. 

Let $Z_i = (z_1, \dots, z_s)$ be a sparse vector, and $x_1, \dots, x_{d'}$ are the indices on which the input vector has non-trivial values for all $1 \leq j \leq d'$.
By the independence of the randomized response mechanism, we observe that \(\mathbb{P}_P[Z_i]/\mathbb{P}_Q[Z_i] = \prod_{j = 1}^{d'} \mathbb{P}_{P_j}[Z_{x_j}]/\mathbb{P}_{Q_j}[Z_{x_j}]\), where \( P_j \) represents the probability distribution of the randomized response result derived from the \( x_j \)-th element of the input vector, and \( Q_j \) corresponds to the one obtained from the reference vector. Since \( z_j \) can be generated independently and is not required for the calculation when \( j \notin \{x_1, \dots, x_{d'}\} \), we can omit \( z_j \) in such cases. Leveraging this observation, we can bypass generating the full vector and instead compute the ratio using only these \( d' \) specific coordinates. The calculation time $c$ is $O(d')$.

\begin{algorithm}
    \caption{Calculate the ratio $\mathrm{d}P / \mathrm{d}Q$ at a given state $Z_i$}
    \label{algo:ratio}

    \Fn{\ProbaRatio}{
  \KwIn{1) A list of indices corresponding to the non-trivial entries of the input vector, denoted as \( (x_1, \ldots, x_{d'}) \); 2) The values of the input vector at these indices, represented as \( (v_1, \ldots, v_{d'}) \); 3) The reference vector values at these indices, given by \( (c_1, \ldots, c_{d'}) \); 4) The value of $Z_i$ at these indices, given by $(z_1, \ldots, z_{d'})$; 5) A privacy budget \( \varepsilon \).}
        \KwOut{$\frac{\mathrm{d}P}{\mathrm{d}Q}(Z_i)$}

        $ratio \gets 1$\;
        \For{$i \gets 1$ \KwTo $d'$}{
            %$B \gets \Gen(Q_i, \mathcal{G}^{(nk+x_i)})$\;
            \lIf{$z_i = v_i$}{
                $ratio \gets ratio \times e^{\varepsilon}$
            }
            \lIf{$z_i = c_i$}{
                $ratio \gets ratio \times e^{-\varepsilon}$
            }
        }
        \Return $ratio$
    }
\end{algorithm}

\section{Theoretical Analysis}
\label{sec:theory}

First, we give the privacy result for our algorithm.

\begin{theorem}
    Algorithm \ref{algo:encode-rr} satisfies $2\alpha\varepsilon$-metric differential privacy.
\end{theorem}

\begin{proof}
Each independent \PPR satisfies \( 2\alpha\varepsilon \)-metric differential privacy, as stated in Theorem~\ref{thm:ppr-privacy}. Since the mechanism partitions the indices into \( m \) groups, its overall privacy guarantee follows from the parallel composition property of differential privacy \cite{mcsherry2009privacy,manurangsi2022differentiallyprivatefairdivision}.
\end{proof}

The number of chunks, \( m \), is a tunable parameter that balances communication cost and execution time. Our analysis focuses on the case where \( m = \beta \varepsilon d \), with \( \beta \) as a parameter. We demonstrate that under this setting, the communication cost is of order \( \bigo{\varepsilon d} \), while the computation cost remains \( \bigo{d} \).

\begin{theorem}
    \label{thm:crr-communication}
    The communication cost of our algorithm is \(\bigo{\varepsilon d}\), where \( d \) represents the number of indices where the input vector differs from the reference vector.
\end{theorem}

\begin{lemma}
    \label{lem:metric-kl}
    For $\mathcal{M}$ a metric differential private mechanism, and $P$ and $Q$ two distributions resulting from $\mathcal{M}$ applied on two datasets at distance $d$, then 
        $\sup_z \frac{\mathrm{d}P}{\mathrm{d}Q}(z) \leq e^{\varepsilon d}$ and 
        $\mathsf{D}(P\|Q) \leq \varepsilon d$.
\end{lemma}

\begin{proof}
Since \( P \) and \( Q \) are derived from the same metric differentially private mechanism applied to two datasets separated by a distance \( d \), it follows that  
$\frac{\mathrm{d}P}{\mathrm{d}Q}(z) \leq e^{\varepsilon d}$, for all $z$.
Using this result, we obtain  
$\mathsf{D}(P \| Q) = \mathbb{E}_{Z \sim P} \left[ \log \left( \frac{\mathrm{d}P}{\mathrm{d}Q}(Z) \right) \right] \leq \varepsilon d$.
\end{proof}

\begin{proof}[Proof of Theorem \ref{thm:crr-communication}]
    The random public permutation function can be obtained through shared common knowledge. Additionally we suppose that all users (including the central server) are aware of the reference vector. This leads us to only consider the sharing of $(K_1, \ldots, K_m)$ for the communication cost.

We define \( d_i = |S_i| \), ensuring that \( d_1 + \cdots + d_m = d \). Each \( d_i \) represents the distance from the reference point of the \(i\)-th vector.
    \PPR is applied independently to every chunk $S_i$.
    Thus, each $K_i$ verifies $\E{\log_2 K_i} \leq \mathsf{D}(P_i\|Q_i) + \frac{\log_2 3.56}{\min\{(\alpha -1)/2, 1\}}$ and $\mathsf{D}(P_i\|Q_i) \leq \varepsilon d_i$ using Lemma~\ref{lem:metric-kl}. This gives a total expected communication cost of
    $\sum_{i=1}^m \left[ \E{\log_2 K_i} + \log_2 (\E{\log_2 K_i} + 1) + 2 \right]        \\ \leq \sum_{i=1}^m \left[ \mathsf{D}(P_i\|Q_i) + O(1) + \log_2 \left(\mathsf{D}(P_i\|Q_i) + O(1) \right)\right]
        \\ \leq  \sum_{i=1}^m \left[\varepsilon d_i + O(1) + \log_2 \left(\varepsilon d_i + O(1) \right) \right] = \bigo{\varepsilon d}$.
\end{proof}

\begin{theorem}
    \label{thm:crr-complexity}
The computational cost of the compressed randomized response proposed in this work is \(\bigo{d}\), where \( d \) represents the number of indices where the input vector differs from the reference vector.
\end{theorem}

\begin{proof}
The primary contributor to computational cost is the runtime of \PPR for the various \( S_i \). Therefore, our analysis will focus on these computations.
For any \( i \in [1, m] \), we first observe that the complexity \( c_i \) of Algorithm~\ref{algo:ratio} applied to \( S_i \) is \(\Theta(d_i)\). Consequently, the expected computational cost \( \mathcal{C}_i \) of running PPR depends on \( d_i \) and is given by  
$\mathcal{C}_i(d_i) = \bigo{e^{\varepsilon d_i} (d_i + \varepsilon d_i)} = \bigo{d_i e^{\varepsilon d_i}}$.

We need to analyze the distribution of \( d_i \) and its impact on the expected value of \( \mathcal{C}_i \). The indices of the vector are randomly shuffled before being divided into \( m \) groups. In the worst-case scenario, the subset \( S_i \) forms a contiguous block of size \( \lceil n / m \rceil \). Under this scenario, \( d_i \) follows a hypergeometric distribution with parameters \( n, d, \lceil n / m \rceil \).
Furthermore, since \( \mathcal{C}_i \) is a convex function, we can apply Theorem~\ref{thm:hypergeometric-inequality}. This result establishes that the expected value of \( \mathcal{C}_i(d_i) \) is upper-bounded by the expected value of \( \mathcal{C}_i(Y) \), where \( Y \sim \mathcal{B}(N, p) \) with \( N = \lceil n / m \rceil \) and \( p = d / n \).

We define the function \( f(x) = x e^{\varepsilon x} \) and compute the expected value of \( f(Y) \). This allows us to bound the complexity of Algorithm~\ref{algo:encode-rr} by \( \bigo{m \cdot \mathbb{E}[f(Y)]} \), where  
$\mathbb{E}[f(Y)] = \sum_{i=0}^{N} i e^{\varepsilon i} p^i (1-p)^{N-i} \binom{N}{i}$.
This expression can be viewed as a function of \( e^{\varepsilon} \), which we denote as \( g \). Additionally, we introduce the function  
$h(x) = \sum_{i=0}^{N} x^i p^i (1-p)^{N-i} \binom{N}{i}$.
From the moment-generating function formula for the binomial distribution, we obtain  
$h(x) = (1 + (x-1) p)^N$.
Differentiating \( h(x) \), we find  
$h'(x) = pN(1 + (x-1) p)^{N-1}$.
Since \( g(x) \) satisfies the relation \( x h'(x) = g(x) \), we conclude that:
    \begin{align*}
        \E{f(Y)} & = e^{\varepsilon} pN \left(1 + (e^{\varepsilon}-1) p \right)^{N-1} \\
        & \leq \frac{d}{n} \left(\frac{n}{m} + 1\right) \left(1 + (e^{\varepsilon}-1) \frac{d}{n} \right)^{\left\lceil\frac{n}{m}\right\rceil - 1} \\
        & \leq \left(\frac{1}{\beta \varepsilon} + \frac{d}{n}\right) \exp \left[\frac{n}{m} \ln \left( 1 + (e^{\varepsilon}-1) \frac{d}{n} \right)\right] \\
        & \leq \left(\frac{1}{\beta \varepsilon} + \frac{d}{n}\right) \exp \left[\frac{n}{m} (e^{\varepsilon}-1) \frac{d}{n} \right] \\
        & = \left(\frac{1}{\beta \varepsilon} + \frac{d}{n}\right) \exp \left( \frac{e^{\varepsilon}-1}{\beta \varepsilon} \right) = \bigo{\frac{1}{\beta \varepsilon}}
    \end{align*}
The final step of the derivation follows from the fact that \(\exp \left( \frac{e^{\varepsilon}-1}{\beta \varepsilon} \right) = O(1)\) when \(\beta \geq 1\) and \(\varepsilon\) is close to zero. Since Algorithm~\ref{algo:encode-rr} executes a total of \( m = \beta \varepsilon d \) instances of \PPR, the overall computational complexity of the algorithm is given by \( \bigo{m \cdot \mathbb{E}[f(Y)]} \), which simplifies to \( \bigo{d} \).
\end{proof}

\section{Potential Applications}
\label{sec:applications}

In this section, we will explore various scenarios in which our algorithm can be applied. Randomized response serves as a fundamental component of differentially private algorithms. Consequently, our mechanism is particularly useful in private applications involving large volumes of data where communication costs are a concern. The only prerequisites are (1) the existence of a reference vector and (2) the adoption of metric differential privacy as the privacy framework. Regarding the first requirement, in most real-world scenarios, the server typically has some prior knowledge of the information held by the user. In the following discussion, we will examine examples of how this prior knowledge about the secret vector can be transformed into a reference vector.

\paragraph{Graph/Social Network Information} The first scenario we examine is the publication of adjacency lists. In this setting, users seek to privately share their list of neighbors in a graph. The primary framework used to protect such information is edge-local differential privacy \cite{qin2017generating}, which aligns with metric differential privacy by defining distance as the Hamming distance between two adjacency vectors. 
Since most real-world graphs' adjacency vectors are sparse, it is generally known in advance that the majority of bits in the adjacency vector will be zeros.  Therefore, we adopt a reference vector consisting entirely of zeros for our mechanism. According to Theorem~\ref{thm:crr-communication}, this choice leads to a communication cost proportional to \( d \), the degree of the node—typically around a thousand—rather than \( n \), the total number of nodes in the graph, which usually reaches several million, as would be required for a naive randomized response.

%Moreover, our approach outperforms even non-private methods, where each of the \( d \) neighbors requires bandwidth on the order of \( \log n \) for communication. This represents a significant contribution, as prior work on communication-efficient publication of adjacency lists has typically introduced substantial noise to mitigate communication costs. In contrast, our mechanism achieves this reduction without the same level of noise increase.

\paragraph{Recommendation System} In this scenario, each user holds a set of ratings for certain items. Those ratings forms a vector called user-item interaction. Specifically, for each item in the set of possible items, a user has either not provided a rating or has assigned a score from a finite set of values. Our goal is to publish these ratings while preserving metric differential privacy, where the distance between two vectors is defined as the number of items for which the ratings differ or are present in only one of the two vectors.

Since the total number of possible items, represented by the size of the user-item interaction vector \( n \), is typically very large, while each user has rated only a small subset, we use an empty rating vector (where no items have been rated) as the reference. Consequently, the number of non-trivial elements in the input vectors, denoted as \( d \), corresponds to the number of items a user has rated, which is significantly smaller than \( n \). This allows our algorithm to achieve a communication cost proportional to the number of rated items rather than the total number of possible items, substantially reducing overhead.

\paragraph{Genomic Information} We focus on the publication of single-nucleotide polymorphisms (SNPs), which represent variations of a single nucleotide in the genome. In this setting, both the central server and users have access to a list of possible SNP locations, along with the most common nucleotide variation at each location. The reference vector is constructed using these most frequent variations, while the input vector represents each user's specific SNP data. Consequently, the vector size \( n \) corresponds to the total number of locations, while \( d \) represents the number of locations where a user's genomic information differs from the most frequent variation. It is known that $d \ll n$ in this type of dataset. 

\section{Experimental Results}
\label{sec:experiments}

In this section, we conduct experiments on the three applications identified in the previous section. The code for these experiments is available in the following repository: \url{https://anonymous.4open.science/r/Metric-DP-Compression-8362}. All timed experiments are conducted on a MacBook Pro (14-inch, 2021) equipped with an M1 Pro chip featuring an 8-core CPU and 32GB of memory. We do not include a comparison with PPR or chunk PPR from \cite{liu2024universal}, as their execution time is prohibitively high across all experimental settings. Additionally, we do not compare the communication cost and execution time with the randomized response technique, as its communication cost is significantly higher than our method and can be theoretically predicted. Moreover, its execution time remains consistently low since it only involves bit flipping.   

\subsection{Recommendation Systems}

We conduct our experiments using the MovieLens 32M dataset \cite{harper2015movielens}, which contains 32 million ratings for 87,585 movies from 200,948 users. For all experiments, we randomly select 1,000 users from this dataset and apply our algorithm to their rating lists. The vector size (representing the number of movies), \( n \), is 87,585, while the number of non-trivial elements (representing the number of ratings per user), \( d \), ranges from a few to 3,500.

\paragraph{Upload Cost and Execution Time}
First, in Figure~\ref{fig:recommender}, we present the communication cost required for users to upload their randomized response to the server, along with the execution time, using the default parameters: \( \varepsilon=1 \), \( \alpha=2 \), and \( \beta=2 \).

The results indicate that the upload cost scales linearly with the number of ratings, as expected, with a slope of approximately 2. The results also confirm that the upload cost of our algorithm is even lower than that of non-private publication of the adjacency vector when using the adjacency list format. While the execution time also increases with the number of ratings, the high variance reduces the strength of this correlation. We have verified that the high variance is not due to the variance of \( d_i \) in our algorithm but rather stems from the PPR. Despite the large variance, all 1,000 executions maintain a manageable execution time.

\begin{figure}[ht]
    \centering
    
    \includegraphics[width=0.4\columnwidth]{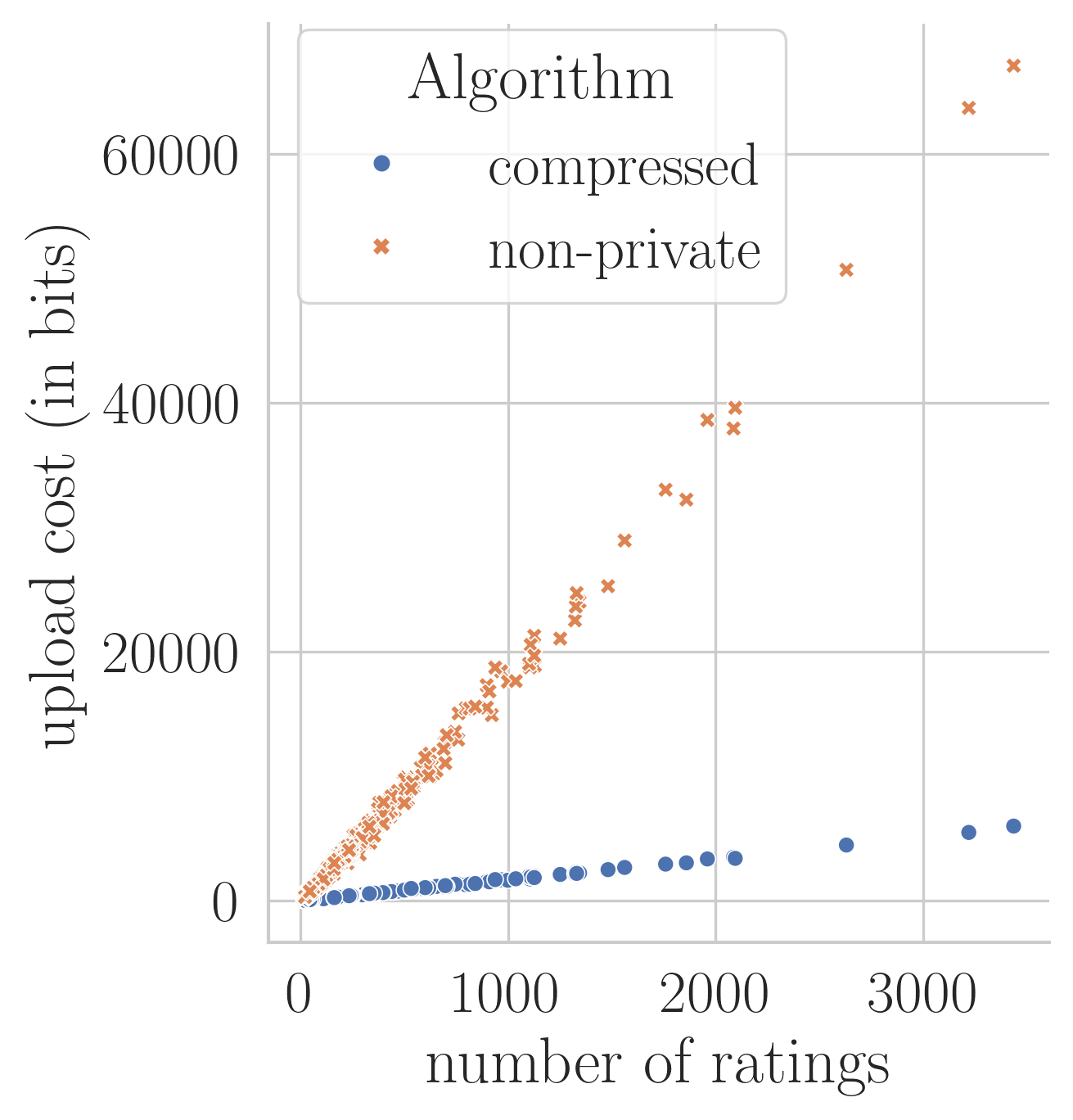}
    \includegraphics[width=0.4\columnwidth]{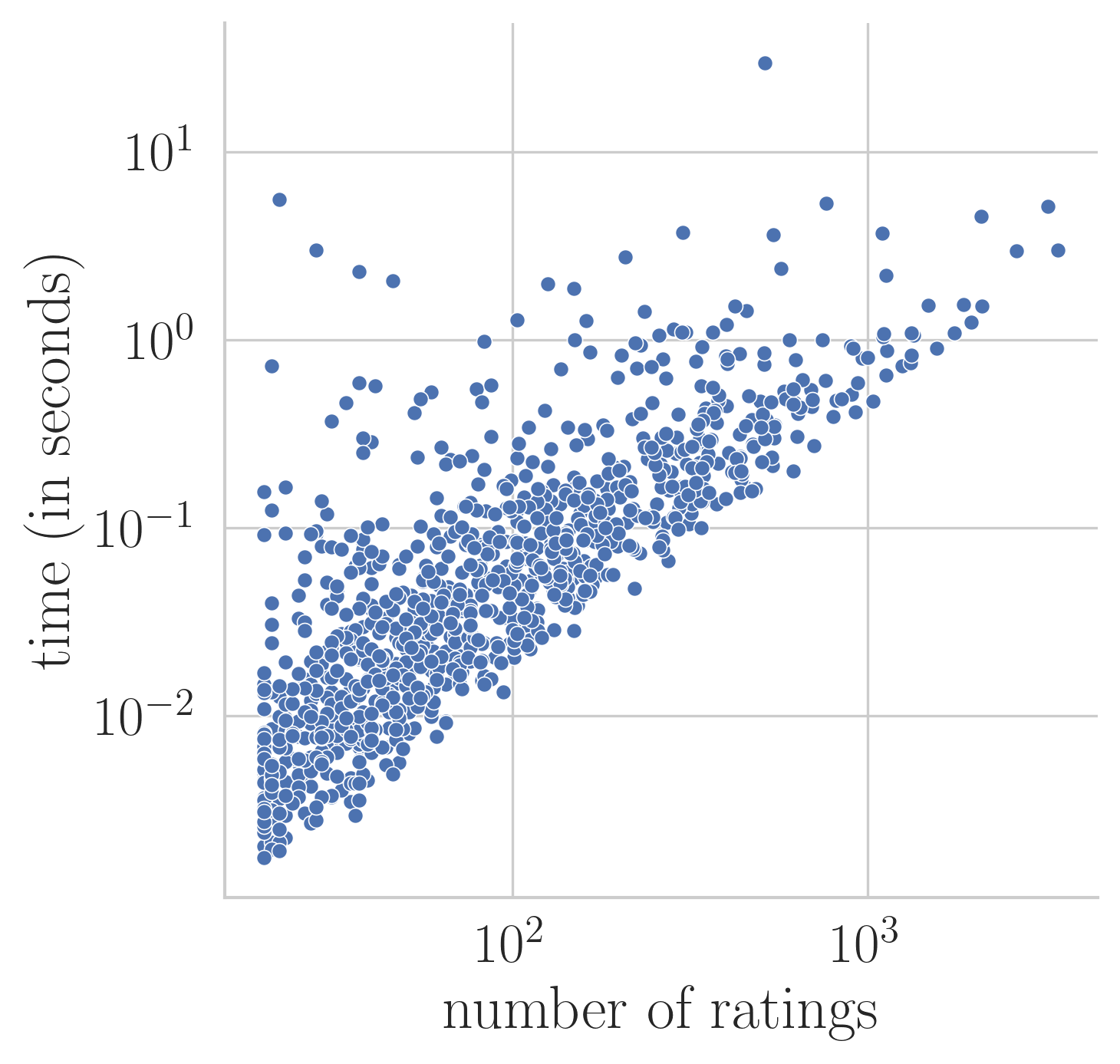}

    \caption{Communication cost and the execution time incurred by our algorithm for 1000 users of the Movie Lens dataset as a function of their number of movie rated}
    \Description{Figure 1. Fully described in the text.}
    \label{fig:recommender}
\end{figure}

\paragraph{Results on Different Privacy Budget $\varepsilon$}
In this experiment, all parameters remain at their default values except $\varepsilon$.  
The results of this experiment are presented in Figure~\ref{fig:recommender-epsilon-upload} for the upload cost and Figure~\ref{fig:recommender-epsilon-time} for the execution time. These results indicate that as the privacy budget increases, both the upload cost and execution time of the algorithm increase, while the variance of the execution time decreases.

\begin{figure}[ht]
    \centering
    \includegraphics[width=0.8\linewidth]{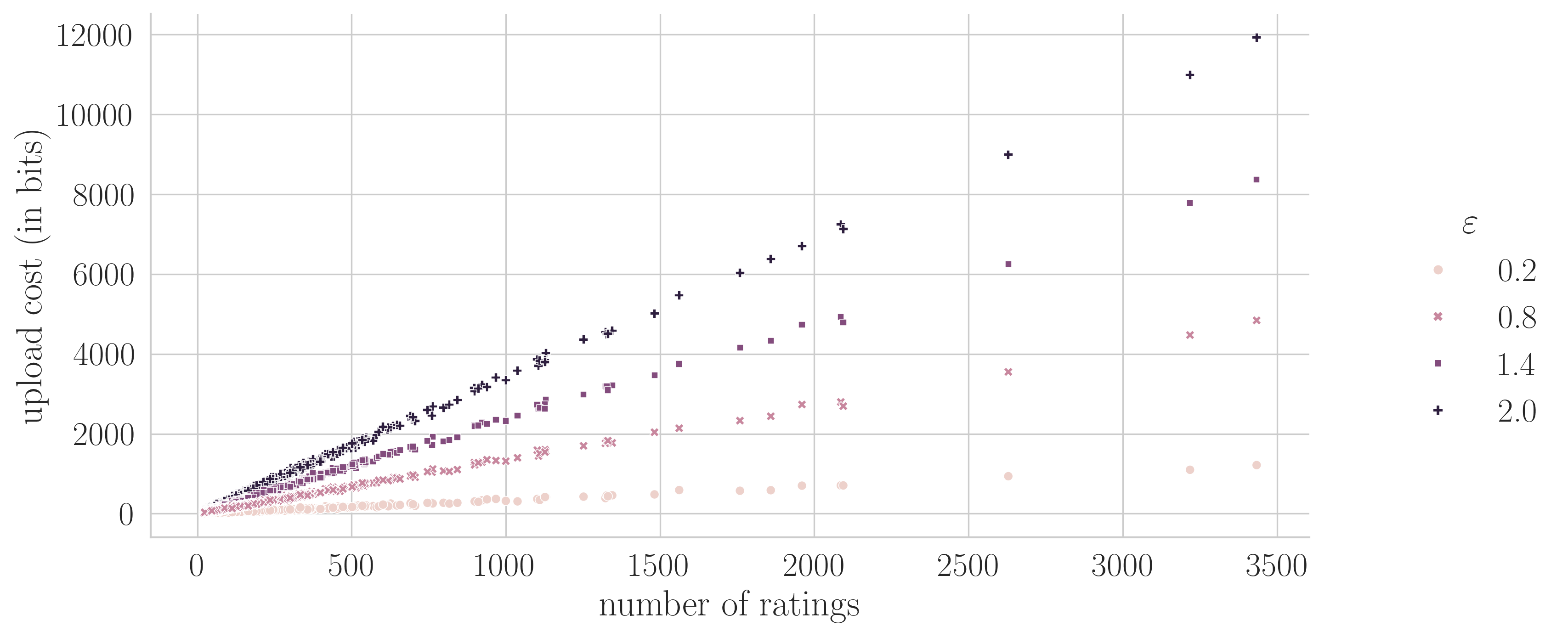}
    \caption{Communication cost of our algorithm for 1,000 users in the MovieLens dataset across different privacy budget values}
    \label{fig:recommender-epsilon-upload}
    \Description{Shows upload cost of the algorithm for 4 different values of privacy budget: 0.2, 0.8, 1.4 and 2.0.}
\end{figure}

\begin{figure}[ht]
    \centering
    \includegraphics[width=0.8\linewidth]{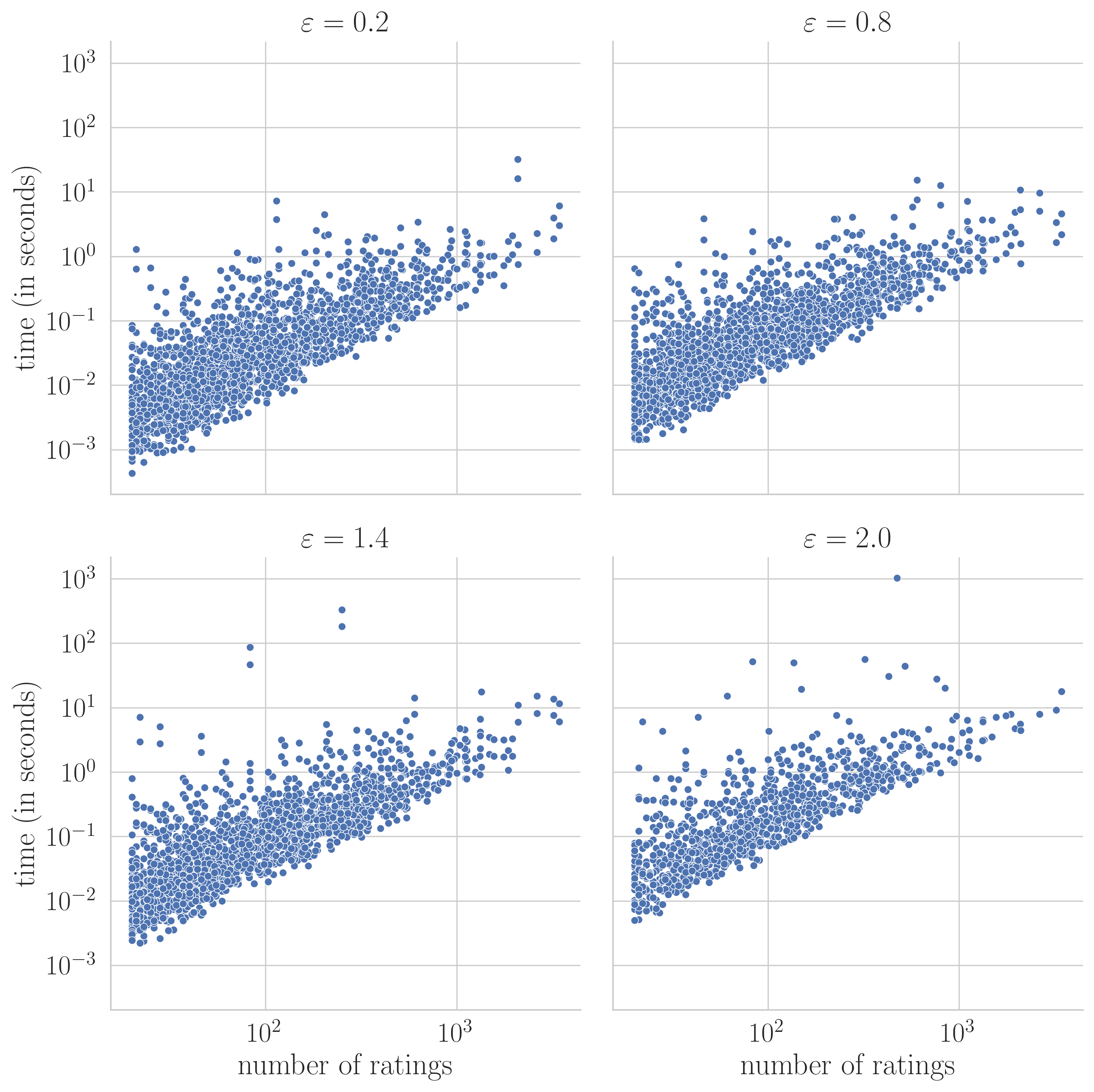}
    \caption{Execution time of our algorithm on 1000 users of the Movie Lens dataset for different values of the privacy budget}
    \label{fig:recommender-epsilon-time}
    \Description{Shows execution time of the algorithm for 4 different values of privacy budget: 0.2, 0.8, 1.4 and 2.0.}
\end{figure}

\paragraph{Results on Different Values of Parameter $\beta$}
Recall that the number of chunks, \( m \), is defined as \( \beta \epsilon d \). In the appendix, we present experimental results for different values of \( \beta \), which confirm that setting the default value to 2 is a reasonable choice.

\paragraph{Estimation of the Number of Common Items}
To prove the accuracy of our method, we also evaluate it on the task of computing the number of common items between 2 users. To this end, we randomly select pairs of users in MovieLens and estimate their number of items in common with classic randomized response and our algorithm, which is called compressed RR in the figure.

\begin{figure}[ht]
    \centering
    \includegraphics[width=0.75\linewidth]{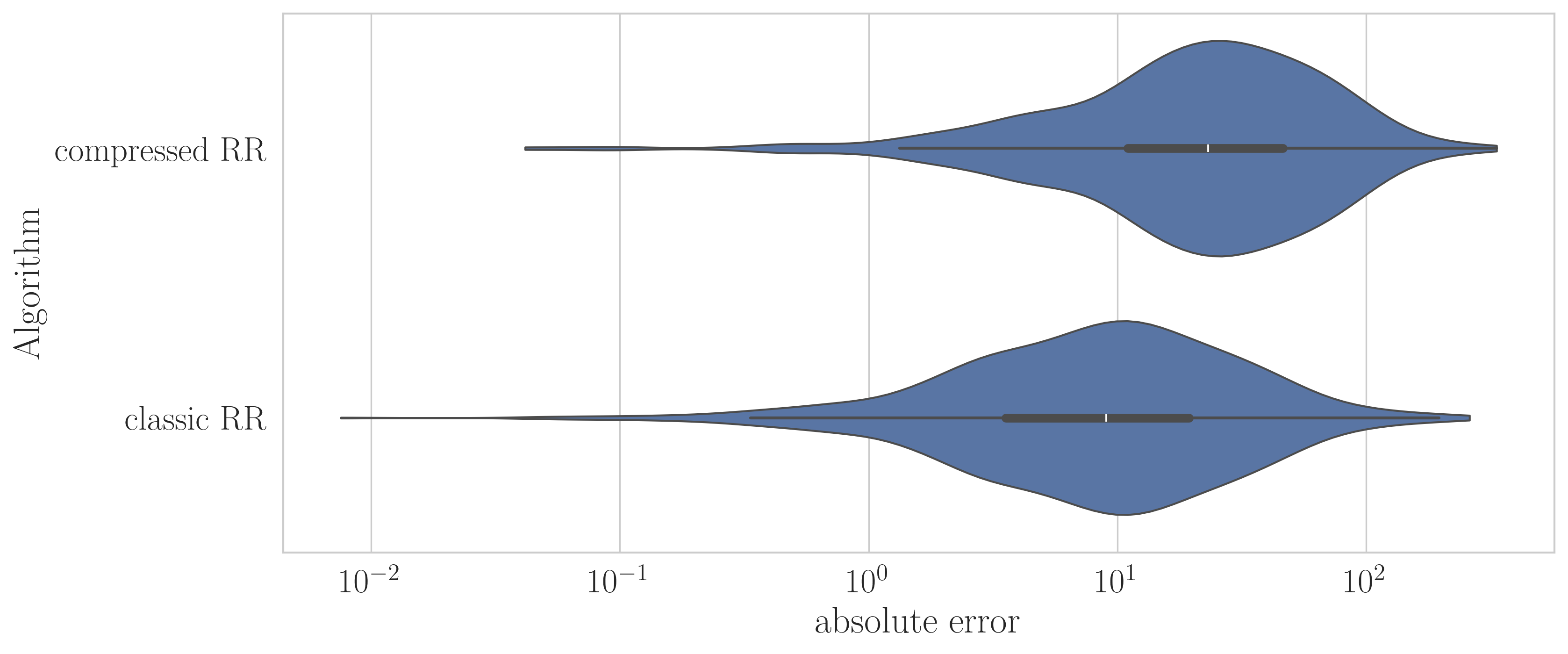}
    \caption{The absolute error in estimating the number of common neighbors across 1,000 user pairs in the MovieLens dataset}
    \label{fig:recommender-distance}
    \Description{Figure 4. Fully described in the text.}
\end{figure}

The results in Figure~\ref{fig:recommender-distance} show that accuracy experiences only a slight decline when using our algorithm, which is expected. Our method can achieve the same results as randomized response but with a larger privacy budget. Consequently, when the budget is fixed, the accuracy is slightly reduced.

\subsection{Genomic Data}

The second application we examine is the publication of SNPs by users. For our experiments, we use chromosome 22 data from the Phase 3 release of the 1000 Genomes Project \cite{10002015global}.

This dataset consists of 1,064,502 locations. While real SNP information is not available due to its sensitivity, we have probability values indicating the likelihood of a user having a variation from the most frequent nucleotide at each location. We use these probabilities to generate synthetic user data. All probabilities are greater than \(10^{-4}\).  
Our method is particularly beneficial for sparse datasets, so we exclude locations where variations are too frequent from our experiments. This exclusion does not pose a limitation, as the principle of parallel composition allows us to publish frequent variations using classical randomized response while applying our algorithm to the rarer ones—without splitting the privacy budget. As a result, we retain only variations with a frequency below \(10^{-2}\), leading to $n = 890,060$ remaining variations, which represents over 83\% of the dataset.

In our experiments, we generate 1,000 synthetic users and apply our algorithm to their SNP lists, with the number of variations ranging from 800 to 1,100. 

In Figure~\ref{fig:dna}, we plot the resulting upload cost and execution time as functions of the number of variations each user possesses. Since the number of variations falls within a narrower range, the distribution is more clearly visible compared to the recommendation system. Notably, the variance in execution time is too high to reveal a clear trend. However, the upload cost remains within a 10\% range of its average value, indicating stable performance. In all of the results, our algorithm exhibits efficient performances both in the upload cost and the execution time.

\begin{figure}[ht]
    \centering
    \includegraphics[width=0.4\columnwidth]{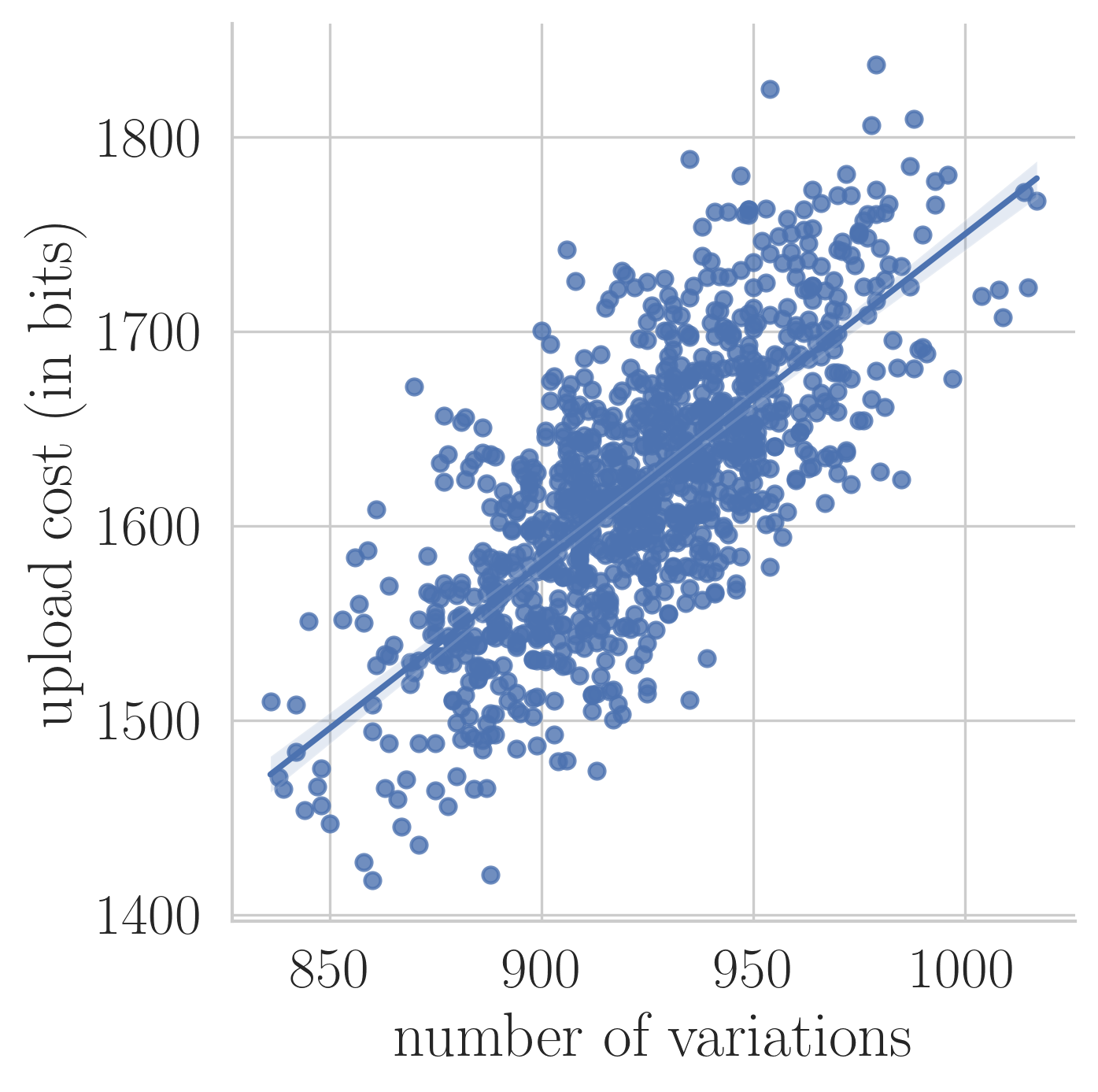}
    \includegraphics[width=0.4\columnwidth]{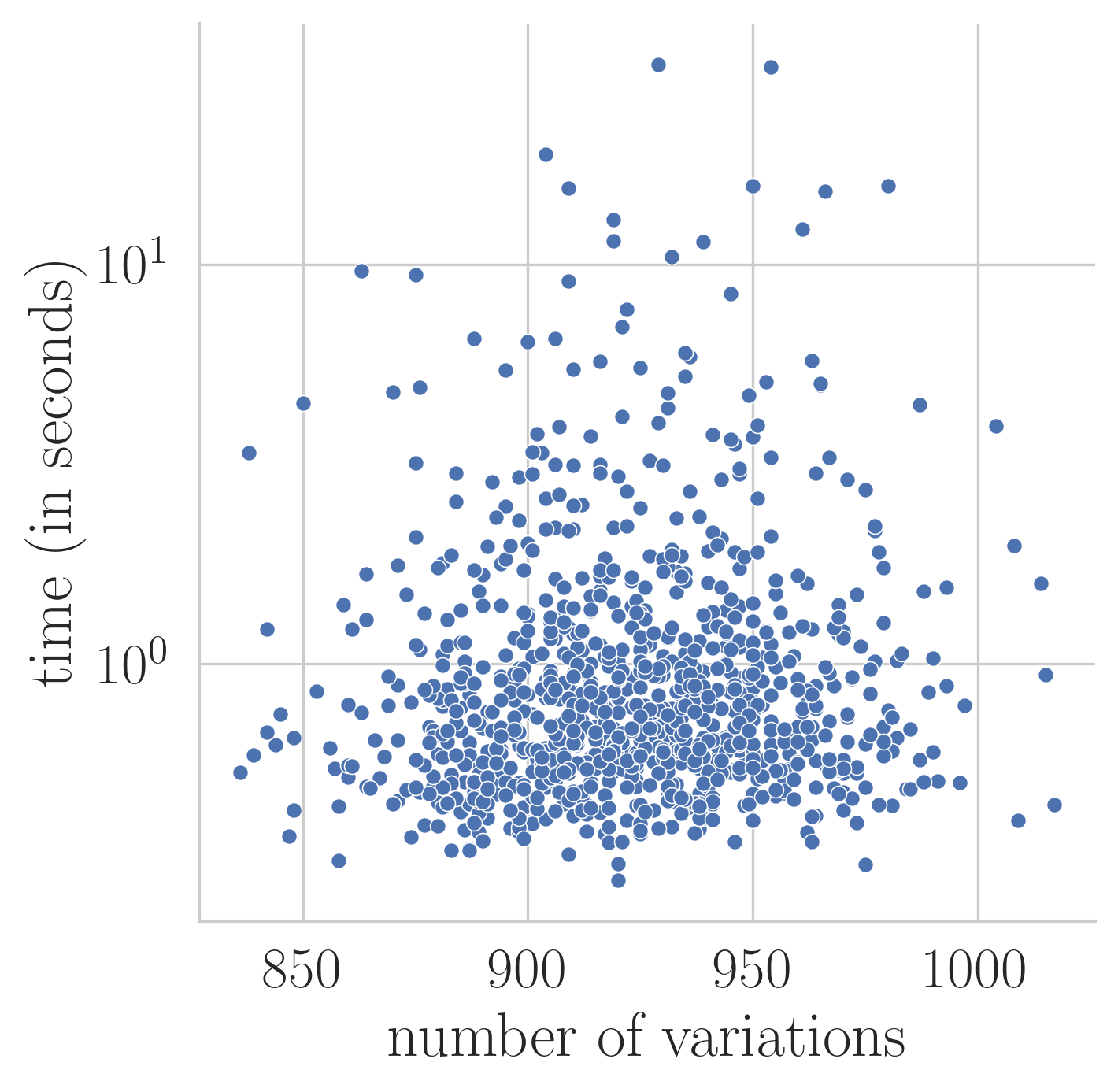}

    \caption{The communication cost and the execution time of our algorithm for 1000 randomly generated SNPs sequences as a function of the number of variations}
    \label{fig:dna}
    \Description{Figure 5. Fully described in the text.}
\end{figure}

\subsection{Social Networks}

\paragraph{Upload Cost and Execution Time} In Appendix, we present our upload cost and execution time on the Google+ dataset, which consists of 107,614 nodes. The results align closely with those observed in the recommendation system and genomic data experiments.

\paragraph{Triangle Counting}
For the experiment on triangle counting we chose to conduct them on the Wikipedia graph \cite{leskovec2010signed,leskovec2010predicting}. This graph contains $n = 7,115$ nodes and 103,689 edges. We are unable to conduct this experiment on the Google+ dataset because the randomized response technique requires excessive memory, making it infeasible to run the algorithm within our computational environment.

We use the two-step mechanisms described in \cite{imola2021locally} to privately estimate the number of triangles in the graph, with one key modification: we replace the classical randomized response with our algorithm.  
In this two-step mechanism, all users must download the randomized response results of every other user to their local storage. As a result, most of the communication cost comes from these download costs. Therefore, unlike other experiments where we compare upload costs, we focus on download costs in this evaluation.

For comparison, we evaluate our method against ARROne \cite{imola2022communication} and GroupRR with CSS \cite{hillebrand2023communication}. These algorithms include a sampling parameter that balances communication cost and accuracy. To demonstrate the full range of their capabilities, we compute their \(\ell_2\)-errors for various values of this parameter.

In contrast, our method does not require such a parameter, so we represent only a single point in the results, corresponding to the average download cost and the average \(\ell_2\)-error, defined as the square root of the sum of squared errors over 10 runs.  
The results are presented in Figure~\ref{fig:triangles}.

\begin{figure}[ht]
    \centering
    \includegraphics[width=0.9\linewidth]{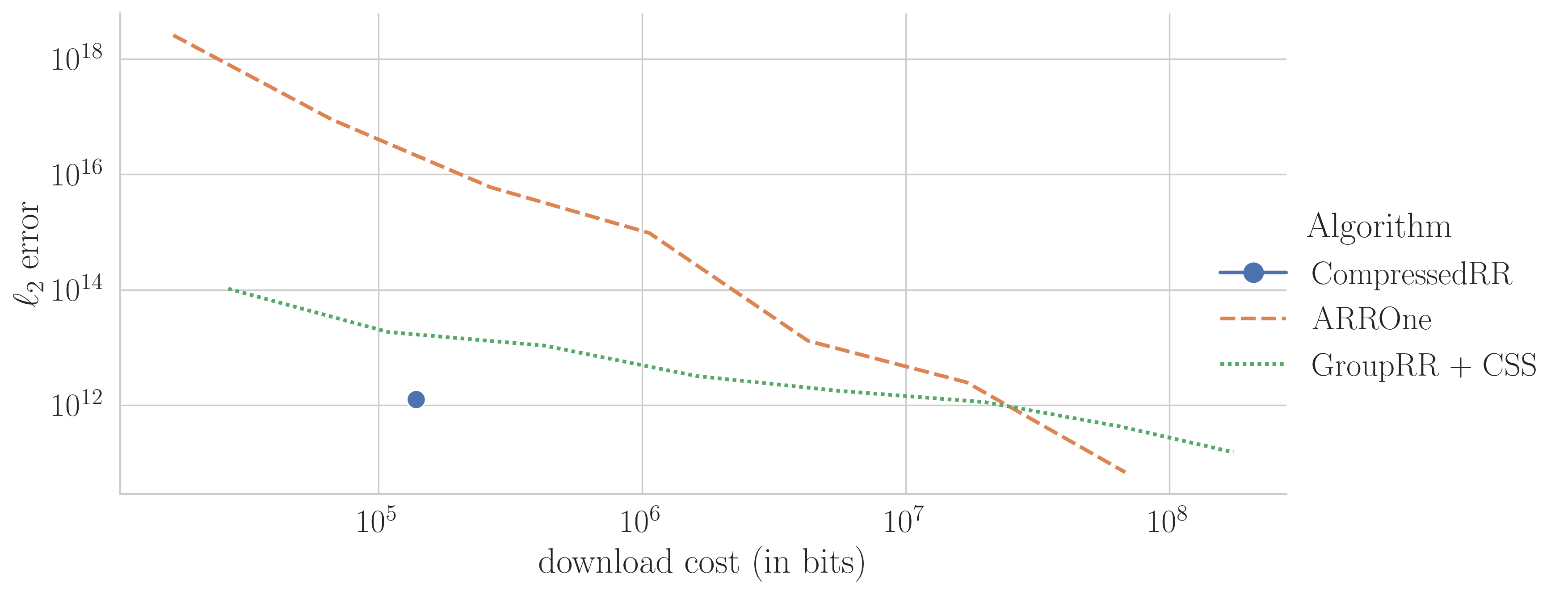}
    \caption{The average $\ell_2$ error in the estimation of the number of triangles in the Wiki graph for 3 different algorithms}
    \Description{Figure 7. Fully described in the text.}
    \label{fig:triangles}
\end{figure}

We observe that, for the same level of accuracy, our method reduces communication cost by a factor of more than 100. Furthermore, at the communication cost used by our method, the error of GroupRR is over 10 times higher, while ARR exhibits an error more than \(10^4\) times worse. We expect these improvements to be even more pronounced for larger or sparser graphs.

\section{Conclusion}

Randomized response is one of the most widely used algorithms for protecting users' sensitive information under metric and local differential privacy, with numerous potential applications. However, for sparse vectors, this method is inefficient in both communication and storage costs.  

Although several studies have proposed solutions to mitigate this issue \cite{imola2022communication,hillebrand2023communication}, we fully resolve it by achieving an even lower cost than non-private communication. This significantly expands the applicability of randomized response. Our algorithm is built on an information-theoretic approach inspired by PPR. While PPR is known to compress information published under differential privacy, its compression rate is approximately \( 1/\epsilon \), where \( \epsilon \) is the privacy budget. In contrast, our compression rate is \( n/(\epsilon d) \), where \( n \) is the vector size and \( d \) is the number of non-trivial values in the sparse vector. This rate is significantly higher than that of PPR.  

Our algorithm is the first to demonstrate how an information-theoretic approach can drastically reduce communication costs in differential privacy applications.

\begin{acks}
    Quentin Hillebrand is partially supported by KAKENHI Grant 20H05965, and by JST SPRING Grant Number JPMJSP2108. Vorapong Suppakitpaisarn is partially supported by KAKENHI Grant 21H05845 and 23H04377. Tetsuo Shibuya is partially supported by KAKENHI Grant 20H05967, 21H05052, and 23H03345.
\end{acks}

% \clearpage

\printbibliography

\clearpage

\section*{Appendix: Additional Experimental Results}

\subsection*{Results on Movie Lens Dataset for Different Values of $\beta$}

In this experiment, all parameters remain at their default values except for \( \beta \). 
The results are presented in Figure~\ref{fig:recommender-beta-upload} for the upload cost and Figure~\ref{fig:recommender-beta-time2} for the execution cost. The findings indicate that increasing \( \beta \) leads to a higher upload cost. For execution time, both its value and variance decrease as \( \beta \) increases. However, this trend is not observed between \( \beta=2 \) and \( \beta=4 \), which led us to select \( \beta=2 \) as the default value.

\begin{figure}[ht]
    \centering
    \includegraphics[width=\linewidth]{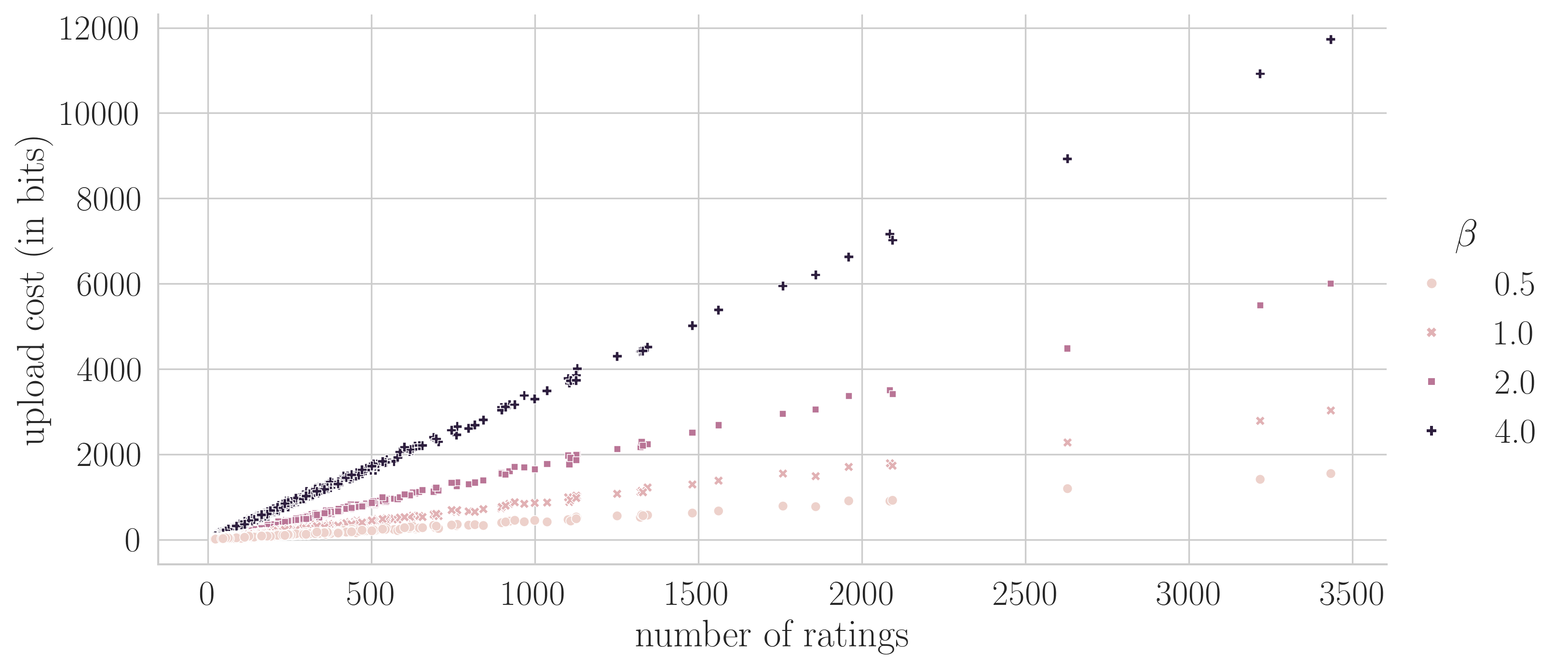}
    \caption{The communication cost of our algorithm on 1000 users of the Movie Lens dataset for different values of the parameter $\beta$.}
    \Description{Shows upload of the algorithm for 4 different values of the parameter beta: 0.5, 1, 2 and 4.}
    \label{fig:recommender-beta-upload}
\end{figure}

\begin{figure}[ht]
    \centering
    \includegraphics[width=\linewidth]{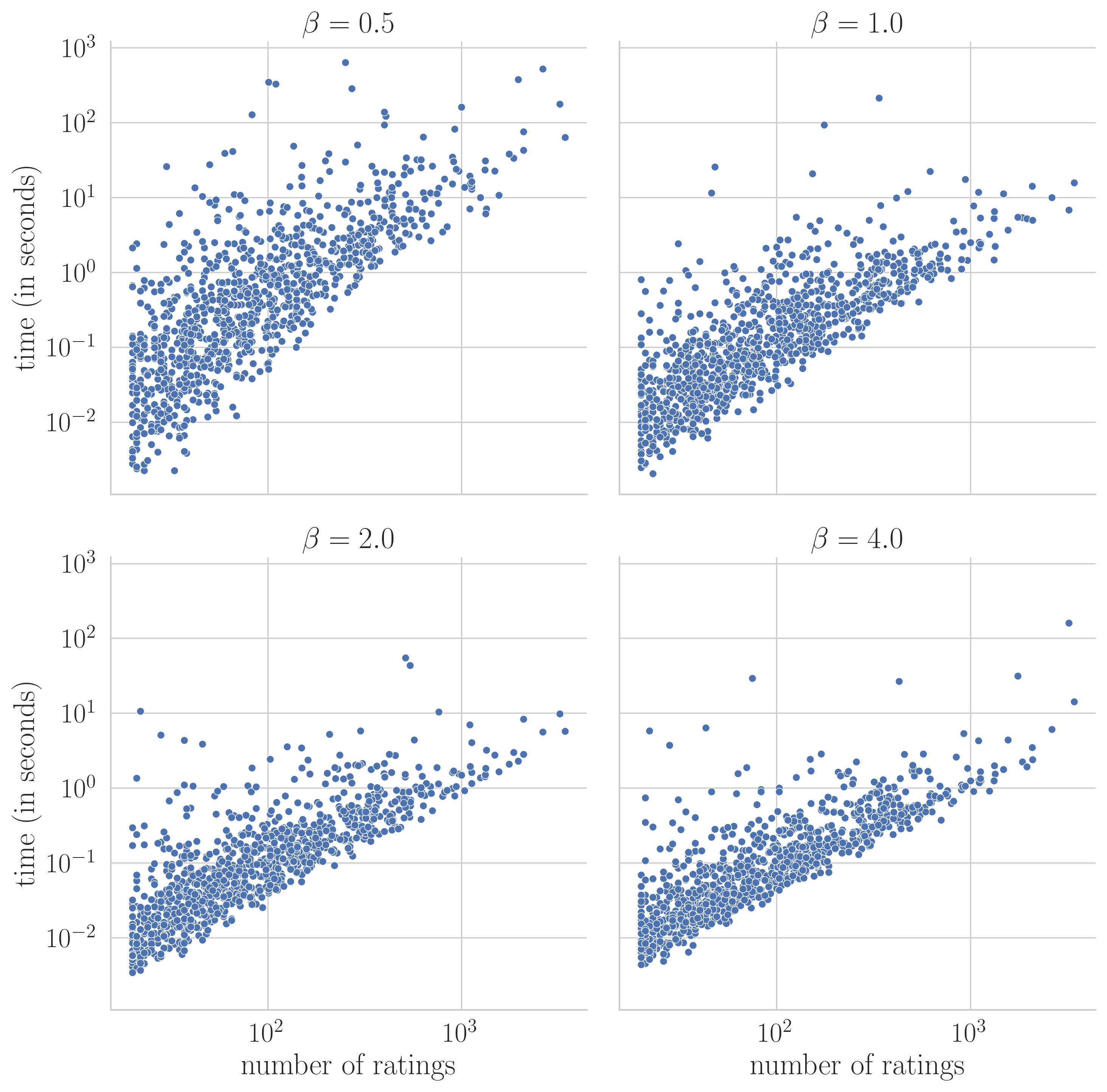}
    \caption{The execution time of our algorithm on 1000 users of the Movie Lens dataset for different values of the parameter $\beta$.}
    \Description{Shows execution time of the algorithm for 4 different values of the parameter beta: 0.5, 1, 2 and 4.}
    \label{fig:recommender-beta-time2}
\end{figure}

\subsection*{Experiments on Google+ Dataset}

We conducted our experiments on the Google+ dataset \cite{leskovec2012learning}, where nodes represent users and edges indicate connections between users within a circle. The resulting graph consists of $n = 107,614$ nodes and 13,673,453 edges. The degree $d$ ranges from 1 to 5,000.

To publish the complete adjacency matrix of an unordered graph, it is sufficient for each user to disclose only their connections with nodes having smaller indices than their own \cite{hillebrand2023communication}. Based on this principle, we applied our algorithm to the adjacency vector, retaining only the 1s corresponding to connections with lower-indexed nodes. Figure~\ref{fig:adjacency} presents the results for 1,000 randomly selected users from the graph.

\begin{figure}[ht]
    \centering
    \includegraphics[width=0.49\columnwidth]{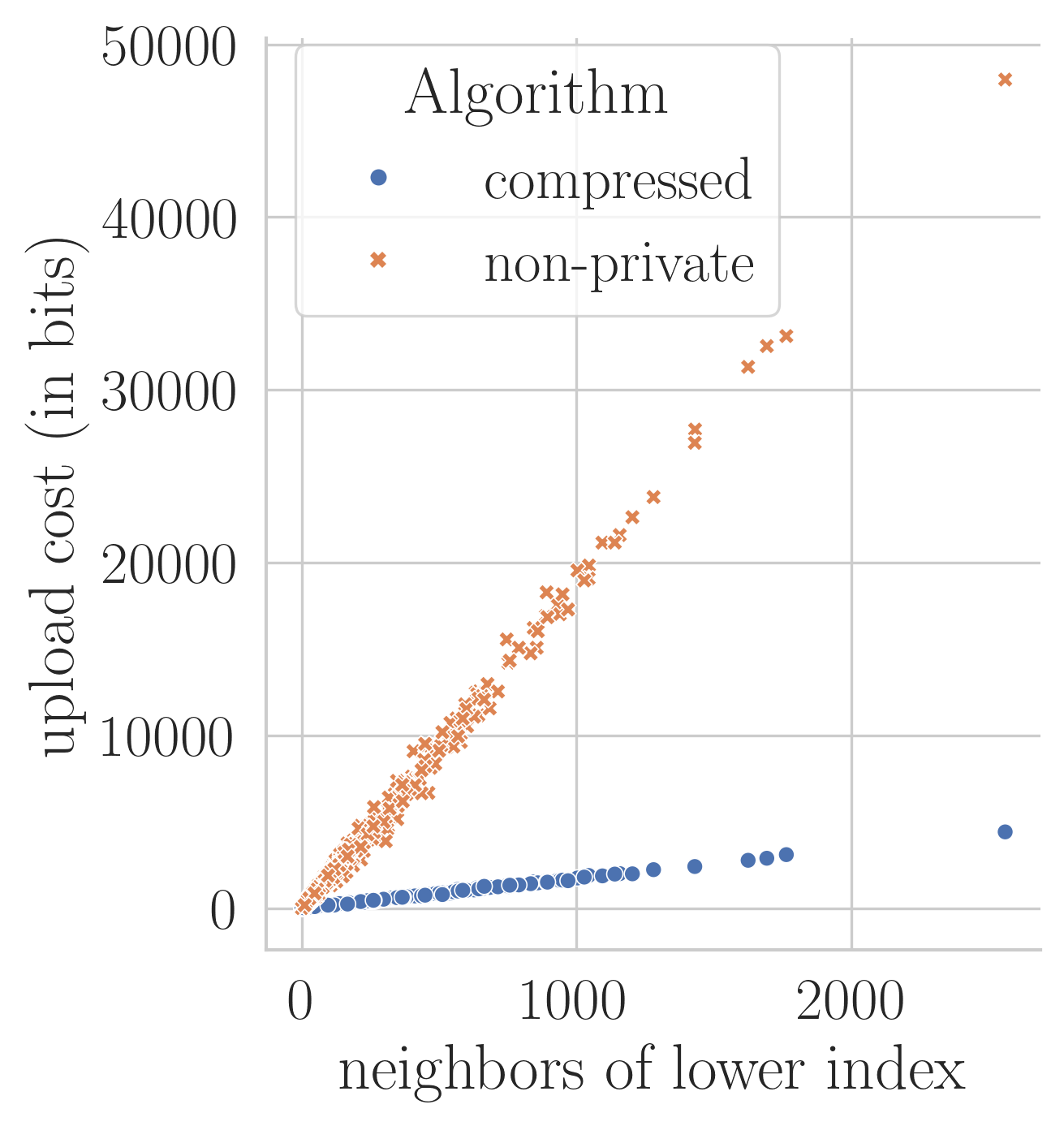}
    \includegraphics[width=0.49\columnwidth]{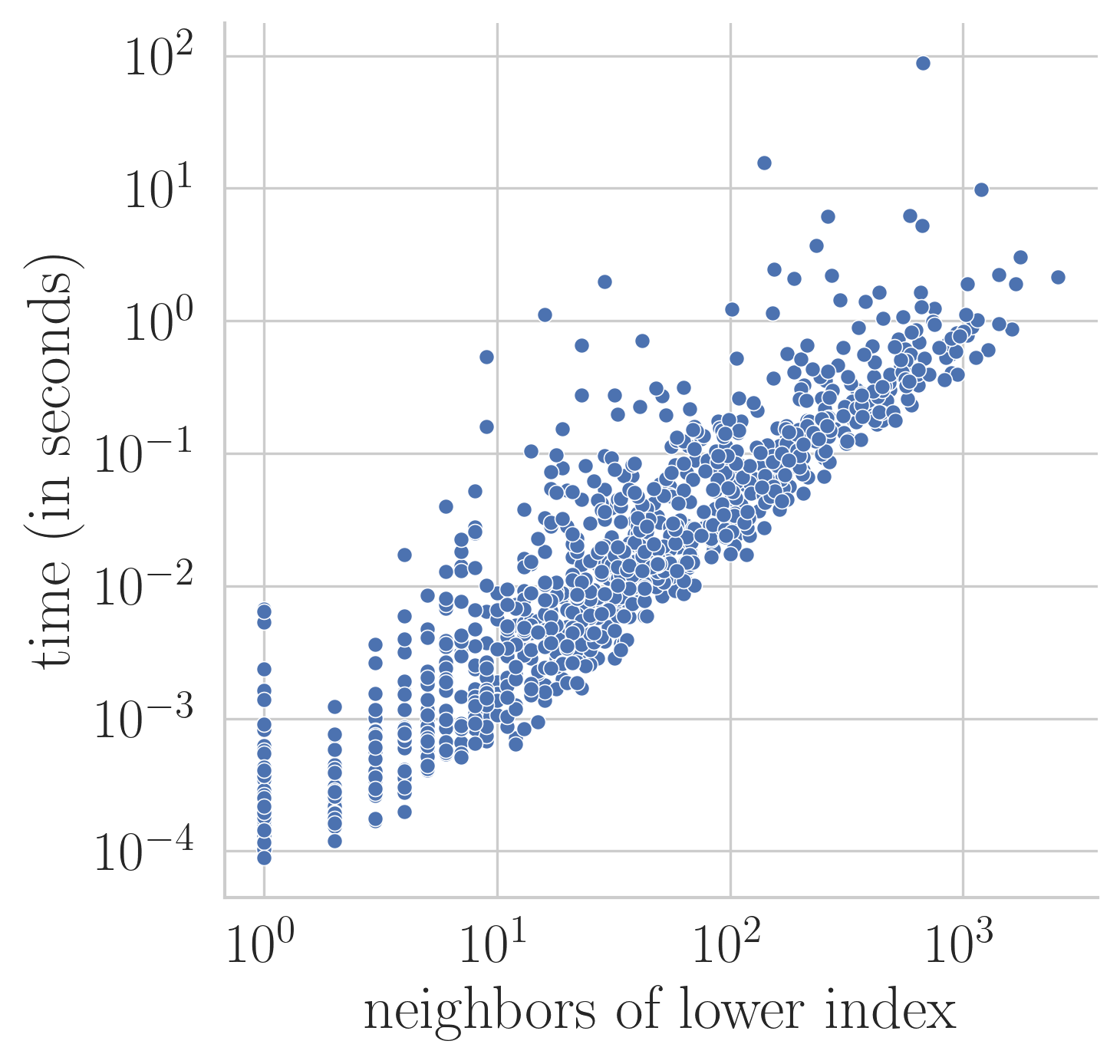}

    \caption{The communication cost and the execution time incurred by our algorithm for 1,000 users of the Google+ dataset as a function of their number of neighbors of smaller index}
    \Description{Figure 6. Fully described in the text.}
    \label{fig:adjacency}
\end{figure}

Similar to user-item interactions in recommendation systems, we observe that both the upload cost and execution time increase with the number of neighbors having smaller indices. Additionally, the variance in execution time is higher than in upload cost. However, we also note that both the upload cost and execution time remain low across all data points.

\end{document}